\documentclass[11pt,reqno]{amsart}
\usepackage{amsmath,amsfonts, amsthm, amssymb}
\usepackage[abbrev,backrefs]{amsrefs}

\usepackage{bm}
\usepackage{microtype}
\usepackage{enumerate}
\usepackage{graphicx}
\usepackage{xcolor}
\usepackage{enumitem} 

\graphicspath{{../../figures/}}
\usepackage[onehalfspacing]{setspace}

\usepackage[foot]{amsaddr}

\DeclareFontFamily{U}{mathx}{\hyphenchar\font45}
\DeclareFontShape{U}{mathx}{m}{n}{
      <5> <6> <7> <8> <9> <10>
      <10.95> <12> <14.4> <17.28> <20.74> <24.88>
      mathx10
      }{}
\DeclareSymbolFont{mathx}{U}{mathx}{m}{n}
\DeclareFontSubstitution{U}{mathx}{m}{n}
\DeclareMathAccent{\widecheck}{0}{mathx}{"71}

\usepackage{scalerel,stackengine}

\def\bra#1{\mathinner{\langle{#1}|}}
\def\ket#1{\mathinner{|{#1}\rangle}}
\def\braket#1{\mathinner{\langle{#1}\rangle}}

\def\Real{{\mathbf R}}

\def\Schwartz{{\mathcal{S}}}

\def\PhaseSpace{{\Real^{2\dm}}}

\def\bbN{{\mathbb N}}
\newcommand{\Wigner}[1]{{\mathcal{W}_{#1}}}
\newcommand{\WignerMinus}[1]{{\mathcal{W}_{#1}^{-}}}
\newcommand{\Husimi}[1]{{\mathcal{Q}^\chi_{#1}}}
\newcommand{\Kernel}[1]{{\mathcal{K}_{#1}}}
\newcommand{\Matel}[1]{{\mathcal{M}^\chi_{#1}}}
\newcommand{\Char}[1]{{\mathcal{F}_{#1}}}

\newcommand{\diff}{\mathop{}\!\mathrm{d}}

\numberwithin{equation}{section}

\newcommand{\tconv}{\circledast}

\DeclareMathOperator{\Tr}{Tr}

\DeclareMathOperator{\trace}{tr}

\newtheorem{theorem}{Theorem}[section]
\newtheorem{lemma}[theorem]{Lemma}
\newtheorem{definition}[theorem]{Definition}
\newtheorem{corollary}[theorem]{Corollary}

\def\Complex{{\mathbf C}}
\newcommand{\rmx}{\mathrm{x}}
\newcommand{\rmp}{\mathrm{p}}

\newcommand{\A}{\rho}
\newcommand{\B}{\eta}
\newcommand{\dm}{n} 
\newcommand{\cz}{} 
\def\Decay{{\mathcal{D}}}
\def\Wigclass{{\mathcal{V}}}

\title[Rapidly decaying Wigner functions]{Rapidly decaying Wigner functions are Schwartz functions}
\author{Felipe Hern{\'a}ndez$^1$}
\address{$^1$Department of Mathematics, Stanford University, 450 Jane Stanford Way, Building 380, Stanford, CA 94305
USA}
\author{C. Jess Riedel$^2$}%
\address{$^2$Physics \& Informatics Laboratories, NTT Research Inc., 940 Stewart Drive, Sunnyvale, CA 94085, USA}
\email{jessriedel@gmail.com}
\date{\today}
\begin{document}
\maketitle
\begin{abstract}
	We show that if the Wigner function of a (possibly mixed) quantum state decays toward infinity faster than any polynomial in the phase space variables $x$ and $p$, then so do all of its derivatives, i.e., it is a Schwartz function on phase space.  
	This is equivalent to the condition that the Husimi function is a Schwartz function, that the quantum state is a Schwartz operator in the sense of Keyl et al., and, in the case of a pure state, that the wavefunction is a Schwartz function on configuration space.
	We discuss the interpretation of this constraint on Wigner functions and provide explicit bounds on Schwartz seminorms.
\end{abstract}

\section{Introduction}

In quantum mechanics, quantum states of $n$ degrees of freedom can be represented by positive semidefinite trace-class operators on $L^2(\Real^{\dm})$. 
Each quantum state $\A$ is  associated with a kernel $\Kernel{\A}$ through $(\A\phi)(x) = \int \Kernel{\A}(x,y)\phi(x)\diff  x$, $\phi\in L^2(\Real^{\dm})$, and the corresponding \emph{Wigner function} $\Wigner{\A}$ is 
\begin{equation*}
\Wigner{\A}(x,p) :=
\frac{1}{(2\pi)^{\dm}}\int e^{ip\cdot y} \Kernel{\A}(x-y/2, x+y/2)\diff y.
\end{equation*}
We denote the set of all such Wigner function as $\Wigclass(\Real^{2\dm})$.
Our main result is a relationship between the decay of such Wigner functions and their 
smoothness.  

To quantify this we use the Schwartz-type seminorms
$|F|_{a,b}$ $:=$ $\sup_{x,p} |x^{a_\rmx} p^{a_\rmp} \partial_x^{b_\rmx} \partial_p^{b_\rmx} F(x,p)|$ of a function $F: \PhaseSpace\to\Complex$ on phase space, with multi-indices $a=(a_\rmx,a_\rmp),b=(b_\rmx,b_\rmp)\in (\mathbb{N}\cup\{0\})^{\times 2\dm}.$
With shorthand notation $|F|_{a}  := |F|_{a,0}$ for the seminorms that only measure the decay of $F$, we say a function is \emph{rapidly decaying} when $|F|_{a\cz}<\infty$ and is a \emph{Schwartz function} when $|F|_{a,b}<\infty$ for all $a,b$.  
We denote the sets of rapidly decaying and Schwartz function by $\Decay(\Real^{2\dm})$ and $\Schwartz(\Real^{2\dm})$, respectively. Our main result:
\newcommand{\mainTheorem}{
    If $\A$ is a positive semidefinite operator whose Wigner function $\Wigner{\A}$ exists and is rapidly decaying, then $\Wigner{\A}$ is a Schwartz function.
}
\begin{theorem}
    \label{thm:main}
    \mainTheorem
\end{theorem}
The assumed rapid decay of $\Wigner{\rho}$ implies $\infty > \int \Wigner{\rho}(\alpha)\diff \alpha=\trace[\rho]$ and hence that $\rho$ is trace-class and so a quantum state. Thus the theorem can be rephrased as the set relation $\Wigclass(\Real^{2\dm}) \cap \Decay(\Real^{2\dm}) \subset \Schwartz(\Real^{2\dm})$.



In this paper we prove Theorem~\ref{thm:main} in two different ways.  The first proof is a bit more abstract, making use of the twisted convolution.  The second proof is a bit more direct, using only basic objects, but requiring more computation. The second proof also results in an explicit bound on the Schwartz seminorms $|\Wigner{\A}|_{a,b}$ of a Wigner function in terms of only its decay seminorms $|\Wigner{\A}|_a$ (Theorem~\ref{thm:seminorm-bound}). 

In the rest of this introduction, we informally sketch the direct (second) proof of Theorem~\ref{thm:main} in order to give the reader intuition, but we stop short of completing the computation.  In the Sec.~\ref{sec:preliminaries}, we recall some notation and basic properties around quantum mechanics in phase space, which can be skipped by experienced readers. In Sec.~\ref{sec:proof} we present our two proofs of our main result and exhibit explicit bounds on the Schwartz seminorms of a Wigner function in terms of its decay seminorms. In  Section~\ref{sec:schwartz-states} we connect our results to the notion of Schwartz operators in the sense of Keyl et al.~\cite{keyl2016schwartz}, and in particular prove the equivalence of a large set of equivalent decay and regularity conditions for various representations of the quantum state. In Sec.~\ref{sec:discussion}, we make some concluding remarks about the ``overparameterization'' of a quantum state by the Wigner function.

\subsection{Motivation}


Why might one think the decay of a Wigner function constrains its derivatives? Consider a pure state $\rho=\ket{\psi}\!\!\bra{\psi}$ with $\ket{\psi}\in L^2(\Real^\dm)$. We can see from the identity
\begin{equation}
    \label{eq:motivation}
    |\widehat{\psi}(p)|^2 = \int \Wigner{\A}(x,p)\diff x
\end{equation}
that rapid decay (in both $x$ and $p$) of the Wigner function implies decay (in $p$) of the Fourier transform $\widehat\psi$ of the wavefunction. This implies that the wavefunction $\psi$ is smooth: $|\psi|_{0,b}<\infty$ for all  $b\in(\mathbb{N}\cup\{0\})^{\times\dm}$. Unfortunately, a bit of trial and error suggests that it is not easy to generalize \eqref{eq:motivation} and obtain a bound on the mixed Schwartz seminorms $|\psi|_{a,b}$ (that is, to show that all the derivatives of $\psi$ are not merely bounded but are also rapidly decaying).

A better way to approach Theorem~\ref{thm:main} avoids privileging either the position or momentum variables by performing a wavepacket decomposition of the quantum state $\rho$.  Using Gaussian wavepackets (coherent states), Zurek argued \cite{zurek2001sub-planck} that if the Wigner function $\Wigner{\A}$ of any quantum state is largely confined to a phase space region of volume $S\sim\ell_\rmx\times\ell_\rmp$, then the smallest structure it will develop is on scales of volume $\Delta s\sim (\hbar/\ell_\rmx)\times(\hbar/\ell_\rmp) \sim\hbar^2/S$. 
This argument was further supported by numerical studies of ``typical'' states generated by chaotic quantum dynamics \cite{zurek2001sub-planck}.




\subsection{Sketch of direct proof}


Consider a family of wavepackets $\chi_\alpha$ of the form
\begin{equation}
    \chi_{(\alpha_\rmx,\alpha_\rmp)}(x) = e^{i(x-\alpha_x/2)\cdot\alpha_p}\chi(x-\alpha_x),
\end{equation}
for fixed smooth envelope function $\chi$ concentrated near the origin. (For example, $\chi$ can be chosen to be a Gaussian.) Given the spectral decomposition of a quantum state
\begin{equation}
    \rho = \sum_j \lambda_j \ket{\psi_j}\!\!\bra{\psi_j},
\end{equation}
we can use the decomposition $\psi_j = (2\pi)^{-\dm} \int \braket{\chi_\alpha|\psi_j} \chi_\alpha \diff\alpha$ for each eigenfunction as an integral over phase space, which is a standard calculation proven in Lemma~\ref{lem:trace}.  We can then express
$\rho$ as
\begin{equation}
    \rho = \frac{1}{(2\pi)^\dm}\sum_j\lambda_j
    \int \ket{\chi_\alpha}\!\!\bra{\chi_\beta} 
    \braket{\chi_\alpha|\psi_j}\braket{\psi_j|\chi_\beta}\diff\alpha\diff\beta,
\end{equation}
Applying the Wigner transform to both sides, this yields a decomposition
\begin{equation}
\label{eq:wigner-decomp-sketch}
    \Wigner{\rho} = \frac{1}{(2\pi)^\dm}
    \int \Wigner{\ket{\chi_\alpha}\!\bra{\chi_\beta}}
    \braket{\chi_\alpha|\rho|\chi_\beta} \diff\alpha\diff\beta.
\end{equation}
in terms of the Wigner transform $\Wigner{\ket{\chi_\alpha}\!\bra{\chi_\beta}}$ of the ``off-diagonal'' operator $\ket{\chi_\alpha}\!\!\bra{\chi_\beta}$. Although $\Wigner{\ket{\chi_\alpha}\!\bra{\chi_\beta}}$ is not a Wigner function (because $\ket{\chi_\alpha}\!\!\bra{\chi_\beta}$ is not positive semidefinite for $\alpha \neq \beta$), it is known \cites{zurek2001sub-planck,toscano2006sub-planck} to be localized near the 
phase space point $(\alpha+\beta)/2$ and has an oscillation with frequency roughly $|\alpha-\beta|$.  

Since $\braket{\chi_\alpha|\rho|\chi_\alpha}$ is just a convolution of the Wigner function $\Wigner{\A}$, the rapid decay of $\Wigner{\A}$ implies the rapid decay of $\braket{\chi_\alpha|\rho|\chi_\alpha}$, and then in turn one can show the rapid decay of $\braket{\chi_\alpha|\rho|\chi_\beta}$ using the Cauchy-Schwartz inequality:
\begin{equation}
\label{eq:rho-CS}
    \braket{\chi_\alpha|\rho|\chi_\beta}^2 \leq
    \braket{\chi_\alpha|\rho|\chi_\alpha}\braket{\chi_\beta|\rho|\chi_\beta},
\end{equation}
which holds because $\rho$ is positive semidefinite.  The assumed decay and smoothness properties of $\chi$ additionally give an estimate of the form 
\begin{equation}
    |\Wigner{\ket{\chi_\alpha}\!\bra{\chi_\beta}}|_{a,b}
    \leq C(a,b) (1+|\alpha|+|\beta|)^{D(a,b)}.
\end{equation}
When combined with decomposition \eqref{eq:wigner-decomp-sketch} of $\Wigner{\A}$, this is enough to show that all the Schwartz seminorms $|\Wigner{\A}|_{a,b}$ are finite.

Our other proof requires additional machinery but still rests heavily on wavepacket decompositions of the quantum state and on the Cauchy-Schwartz inequality~\eqref{eq:rho-CS}.

\section{Preliminaries}
\label{sec:preliminaries}

This section establishes our notation and reviews standard features of phase-space representations of quantum mechanics. (To keep this paper self-contained, we provide proofs of the lemmas in this section in the Appendix.) Throughout this 
paper, we take $\chi\in\Schwartz(\Real^\dm)$ to be a fixed Schwartz function that is normalized, $\|\chi\|_{L^2(\Real^\dm)} = \int |\chi(y)|^2\diff y=1$, but otherwise arbitrary.\footnote{A standard choice is to specialize to a Gaussian coherent state $\chi(y) = \exp(-x^2/2)/\sqrt{(2\pi)^\dm}$ (especially when used as in Subsection~\ref{sec:displacement} as the reference wavefunction with respect to which Husimi function is defined). However, this specialization is not necessary and one could instead take $\chi$ to be, e.g., a smooth and compactly supported wavefunction.}

Experienced readers may prefer to skip directly to Sec.~\ref{sec:proof} for the proof of our main result and only refer back to this section as necessary.

\subsection{Notation}
\label{sec:notation}

In what follows, a \emph{wavefunction} of $\dm$ continuous quantum degrees of freedom is represented by a member of $L^2(\Real^{\dm})$ and denoted by $\psi$, $\phi$, or $\chi$.  A \emph{quantum state} is the possibly mixed generalization, represented by a positive semidefinite (and hence self-adjoint) trace-class operator on $L^2(\Real^{\dm})$ and denoted by $\A$ or $\B$.  
Vectors on phase space are $\alpha,\beta,\gamma, \xi\in\PhaseSpace$ with position and momentum components denoted by (for example) $\alpha_\rmx,\xi_\rmp\in\Real^\dm$. Multi-indices are $a,b,c,d\in(\mathbb{N}\cup\{0\})^{\times 2\dm}$ (or ($\mathbb{N}\cup\{0\})^{\times\dm}$ in Sec.~\ref{sec:schwartz-states}) with 
$\alpha^b = \alpha_\rmx^{b_\rmx}\alpha_\rmp^{b_\rmp} = \prod_{i=1}^{2n}\alpha_i^{b_i}$, 
$|b|=|b_\rmx|+|b_\rmp|=\sum_{i=1}^{2n} b_i$, 
$b!=b_\rmx ! b_\rmp! = \prod_{i=1}^{2n} b_i !$, 
and $\binom{a}{b} = a!/((a-b)!b!)$. We use $b\le c$ to mean $b_i \le c_i$ for all $i = 1, 2, \ldots 2\dm$.

The symplectic form is $\alpha\wedge\beta = \alpha\cdot\Omega\cdot\beta=\alpha_\rmx\cdot \beta_\rmp - \alpha_\rmp\cdot \beta_\rmx$, with $\Omega=\left(\begin{smallmatrix}0 & I\\ -I & 0\end{smallmatrix}\right)$ an antisymmetric matrix on $\Real^{2\dm}$, $I$ the identity matrix on $\Real^\dm$, and ``$\cdot$'' the dot product on $\Real^\dm$ and $\Real^{2\dm}$. The position and momentum operators are $X=(X_1,\ldots,X_\dm)$ and $P=(P_1,\ldots,P_\dm)$, which are combined into the phase-space operator $R=(X,P)$.  For a given quantum state $\A$ and reference wavefunction $\chi\in\Schwartz(\Real^\dm)$, some associated functions over phase space, doubled phase space, and doubled configuration space are $\Wigner{\A}$, $\Husimi{\A}$, $\Matel{\A}$, $\Char{\A}$, and $\Kernel{\A}$ (defined below). We use ``$*$'' to denote the convolution, $(f * g)(\alpha) = \int f(\alpha-\beta) g(\beta)\diff \beta$.
Given a matrix form $\Omega'$ we also define the twisted convolution
\begin{equation}
    (f \tconv_{\Omega'} g)(\alpha)
    = \int e^{i\alpha\cdot \Omega'\cdot \beta/2} f(\alpha-\beta)g(\beta)\diff\beta.
\end{equation}


For any wavefunction $\phi\in L^2(\Real^{\dm})$, we denote the linear functional associated with it using bra notation, $\bra{\phi}=(\psi \mapsto \int \bar{\phi}(x)\psi(x)\diff x) \in \Schwartz'(\Real^{\dm})$, and denote the scalar result with a bra-ket, $\braket{\phi|\psi} = \bra{\phi}(\psi) = \int \bar{\phi}(x)\psi(x)\diff x$. 
More generally, with an operator $E$ we write $\braket{\phi|E|\psi} = \bra{\phi}(E\psi)=\bra{E^\dagger\phi}(\psi)$.  For any two wavefunction $\phi_1,\phi_2\in L^2(\Real^\dm)$, we use $\ket{\phi_1}\!\!\bra{\phi_2}$ for the rank-1 operator $\psi \mapsto \braket{\phi_2|\psi} \phi_1$.

\subsection{The displacement operator and phase-space functions}
\label{sec:displacement}

In this subsection we recall standard results about quantum mechanics in phase space (see, e.g., Chapter 1 of Ref.~\cite{folland1989harmonic}).  



\begin{definition}
    \label{def:displacement}
    For $\xi\in\PhaseSpace$, define the \emph{(Weyl generator) displacement operator}
    \begin{equation}
    D_\xi := e^{i\xi\wedge R} = e^{i(\xi_\rmx \cdot P - \xi_\rmp \cdot  X)}.
    \end{equation}
\end{definition}

The following lemma describes the action of $D_\xi$
on an arbitrary wavefunction.
\newcommand{\displacementlemma}{
    For any $\phi\in L^2(\Real^\dm)$,
    \begin{equation}
    D_\xi\phi(y) = e^{i(y-\xi_\rmx/2)\cdot \xi_\rmp} \phi(y-\xi_\rmx).
    \end{equation}
}
\begin{lemma}
\label{lem:displacement}
\displacementlemma
\end{lemma}
It's easy to check these basic properties: $D_\alpha D_\beta = e^{i\beta\wedge\alpha/2}D_{\alpha+\beta}$ and $D_\alpha^\dagger = D_{-\alpha}$.

Now we introduce the quasicharacteristic function, the Wigner function, and the Kernel.
\begin{definition}
    \label{def:characteristic-wigner-kernel}
    For a given quantum state $\A$, the \emph{quasicharacteristic function} is
    \begin{equation}
    \label{char-def}
    \Char{\A}(\xi) := \trace[\A D_\xi].
    \end{equation}
    where the trace is well defined because $\A$ is trace-class and $D_\xi$ is a bounded operator on $L^2(\Real^\dm)$.  
    Because a quantum state $\A$ is necessarily compact, it has a spectral decomposition \cite{gohberg2000trace}
    \begin{equation}
    \begin{split}
    (\A \phi)(x) &= \sum_{i=1}^\infty  \psi_i(x)\braket{\psi_i|\phi}
    \end{split}
    \end{equation}
    with unnormalized eigenvectors $\psi_i\in L^2(\Real^\dm)$ and associated kernel $\Kernel{\A}$ satisfying $(\A\phi)(x) = \int \Kernel{\A}(x,y)\phi(x)\diff x$ and
    \begin{equation}
        \Kernel{\A}(x,y) = \sum_{i=1}^\infty \psi_i(x)\bar{\psi}_i(y) 
    \end{equation}
    almost everywhere.
    Finally, we define the \emph{Wigner function} of $\A$ as
    \begin{equation}
    \label{wigner-traditional-def}
    \Wigner{\A}(\alpha) := \frac{1}{(2\pi)^\dm}\int e^{i\alpha_\rmp\cdot y} \Kernel{\A}(\alpha_\rmx-y/2, \alpha_\rmx+y/2)\diff y,
    \end{equation}
    where $\Wigner{\A}\in L^2(\Real^{2\dm})$ because it is a Fourier transform of $\Kernel{\A}\in L^2(\Real^{2\dm})$ in one variable. 
\end{definition}

More generally, we call $\Wigner{E}(\alpha) := (2\pi)^{-\dm}\int e^{i \alpha_\rmp\cdot y} \Kernel{E}(\alpha_\rmx-y/2, \alpha_\rmx+y/2)\diff y$ and $\Char{E}(\xi) := \trace[E D_\xi]$ the \emph{Wigner transform} and \emph{quasicharacteristic transform} of any kernel operator $E$, which in particular exists for any rank-1 operator $E=\ket{\phi}\!\!\bra{\psi}$ since $\Kernel{\ket{\phi}\!\bra{\psi}} \in L^2(\Real^{2\dm})$.

\begin{lemma}
    \label{lem:wigner-char-relation}
     For any trace-class kernel operator $E$, the corresponding Wigner transform and quasicharacteristic transform are symplectic Fourier duals:
    \begin{equation}
     \Wigner{E}(\alpha) = \frac{1}{(2\pi)^{2\dm}}\int e^{-i\alpha\wedge\xi}\Char{E}(\xi)\diff\xi.
    \end{equation}
\end{lemma}
The preceding expression is sometimes used as the definition of the Wigner transform, and it is notable for manifestly respecting the symplectic structure of phase space.  The perhaps more traditional definition \eqref{wigner-traditional-def} relies on the kernel, and hence privileges position over momentum, but has the advantage of being more obviously well-defined. 

\newcommand{\wignerKernelSchwartzLemma}{The Wigner function $\Wigner{\A}$ is a Schwartz function if and only if the kernel $\Kernel{\A}$ is a Schwartz function.}
\begin{lemma}
    \label{lem:wigner-kernel-schwartz}
	\wignerKernelSchwartzLemma
\end{lemma}
Roughly speaking, this is because
$\Wigner{\A}$ and $\Kernel{\A}$ are Fourier transforms of each other in \emph{one} of their two variables (after the linear change of variables $(x,y)\to(\bar{x}=(x+y)/2,\Delta x = x-y)$).  

\newcommand{\twistedConvolutionLemma}{The twisted convolution of a rapidly decaying function with a Schwartz function is itself a Schwartz function.}
\begin{lemma}
    \label{lem:twisted-convolution}
	\twistedConvolutionLemma
\end{lemma}
The proof is essentially the same as for the similar statement with the normal convolution.

\newcommand{\hilbertschmidtwignerlemma}{
    For any two quantum states $\A$ and $\B$,
    \begin{align}
    \trace[\A\B]=(2\pi)^{\dm} \int\Wigner{\A}(\alpha)\Wigner{\B}(\alpha)\diff \alpha.
    \end{align}
}
\begin{lemma}
    \label{lem:hilbert-schmidt-wigner}
    \hilbertschmidtwignerlemma
\end{lemma}

Now we introduce the Husimi function and the so-called matrix element; these are most often defined with respect to a preferred Gaussian reference wavefunction, but we will allow more generality (see, e.g., Ref.~\cite{klauder2007generalized}).

\begin{definition}
    Fixing a reference wavefunction $\chi\in\Schwartz(\Real^\dm)$ that is normalized ($\|\chi\|_{L^2(\Real^\dm)} = \int |\chi(y)|^2\diff y=1$), and Schwartz-class but otherwise arbitrary, we define the \emph{matrix element}
    \begin{equation}
    \label{matel-def}
    \Matel{\A}(\alpha,\beta) :=
    \braket{\chi_\alpha |\A|\chi_\beta},
    \end{equation}
    and the \emph{Husimi function}
    \begin{equation}
    \label{husmi-traditional-def}
    \Husimi{\A}(\alpha) :=
    \braket{\chi_\alpha |\A|\chi_\alpha} = \Matel{\A}(\alpha,\alpha).
    \end{equation}
    using the shorthand $\ket{\chi_\alpha}:=D_\alpha\ket{\chi}$.
\end{definition}

\begin{lemma}
    \label{lem:trace}
    For any trace-class operator $E$ and any $\chi\in L^2(\Real^\dm)$ satisfying $\|\chi\|_{L^2(\Real^\dm)}=1$, 
    \begin{align}
    \trace[E]=\frac{1}{(2\pi)^\dm}\int \braket{\chi_\alpha|E|\chi_\alpha}\diff \alpha
    \end{align}
In particular, for any $\phi,\psi\in L^2(\Real^\dm)$
    \begin{align}
    \label{rank-1-trace}
    \braket{\phi|\psi}=\frac{1}{(2\pi)^\dm}\int \braket{\phi|\chi_\alpha}\braket{\chi_\alpha|\psi}\diff \alpha
    \end{align}
\end{lemma}

\newcommand{\husimiconvolutionlemma}{
    For any quantum state $\A$ and reference wavefunction $\chi\in\Schwartz(\Real^\dm)$,
    \begin{align}
    \Husimi{\A}(\alpha) &= (2\pi)^{\dm} (\Wigner{\A}\ast\WignerMinus{\chi})(\alpha) 
    = (2\pi)^{\dm}\int \Wigner{\A}(\beta)\Wigner{\chi}(\beta-\alpha)\diff\beta
    \end{align}
    where ${\WignerMinus{\chi}}(\alpha):=\Wigner{\chi}(-\alpha)$ is a Schwartz function.
}
\begin{lemma}
\label{lem:husimi-convolution}
    \husimiconvolutionlemma
\end{lemma}

\section{Proof that rapidly decaying Wigner functions are Schwartz function}
\label{sec:proof}

The first (more abstract) proof of our main result is given in subsection \ref{sec:abstract-proof} below. The second (more direct) proof follows in subsection \ref{sec:direct-proof}. These two subsections are independent of each other and can be read in either order.

Both proofs will make crucial use of the Cauchy-Schwartz inequality in the following form: 

\begin{lemma}
\label{lem:husimi-bound-matel}
For any quantum state $\A$, the Husimi function bounds the matrix element:
\begin{equation}
    |\Matel{\A}(\alpha,\beta)|^2 
    \le \Husimi{\A}(\alpha)\Husimi{\A}(\beta)
\end{equation}
\end{lemma}
\begin{proof}
We have 
\begin{equation}
    |\Matel{\A}(\alpha,\beta)|^2 = |\braket{\chi_\alpha|\A|\chi_\beta}|^2
    \le \braket{\chi_\alpha|\A|\chi_\alpha}\braket{\chi_\beta|\A|\chi_\beta} = \Husimi{\A}(\alpha)\Husimi{\A}(\beta)
\end{equation}
where the inequality is just the Cauchy-Schwartz inequality
applied to the inner product $\langle \phi_1,\phi_2\rangle_\A := \braket{\phi_1|\A|\phi_2}$.
\end{proof}
\begin{corollary}
    \label{cor:rapid-decay}
    If the Wigner function $\Wigner{\A}$ of a quantum state $\A$ is rapidly decaying, then the Husimi function $\Husimi{\A}$ and the matrix element $\Matel{\A}$ are also rapidly decaying.
\end{corollary}
\begin{proof}
By Lemma~\ref{lem:husimi-convolution}, the Husimi function $\Husimi{\A}$ is a convolution of the rapidly decaying $\Wigner{\A}$ by the Schwartz function $\WignerMinus{\chi}(\alpha)=\Wigner{\chi}(-\alpha)$, so $\Husimi{\A}$ must also be rapidly decaying. We then get rapid decay of $\Matel{\A}$ using Lemma~\ref{lem:husimi-bound-matel}.
\end{proof}

We now turn to the first strategy.

\subsection{Abstract proof}
\label{sec:abstract-proof}

Here is a sketch of our strategy: We obtain a reproducing formula expressing $\Matel{\A}$ as a twisted convolution of itself with a Schwartz function constructed from $\chi$, showing that $\Matel{\A}$ must itself be a Schwartz function. Then we find an integral expression for the Wigner function $\Wigner{\A}$ in terms of the matrix element $\Matel{\A}$, from which it follows that $\Wigner{\A}$ is a Schwartz function.





\begin{lemma}
    \label{lem:reproducing}
    Let $\Omega'=\left(\begin{smallmatrix}\Omega & 0 \\ 0 & -\Omega\end{smallmatrix}\right)$ be a symplectic form
    on $\Real^{4n}$.  Then the matrix element $\Matel{\A}$ satisfies the following reproducing formula
    \begin{equation}
    \label{reproduce-eq}
    \Matel{\A} = \frac{1}{(2\pi)^{2\dm}}(\Char{\chi}\bar\otimes \Char{\chi})\tconv_{\Omega'} \Matel{\A}
    \end{equation}
    where $\Char{\chi} = \Tr[\ket{\chi}\!\!\bra{\chi}D_\xi] = \braket{\chi|D_\xi|\chi} = \braket{\chi_{-\xi/2}|\chi_{\xi/2}} $ is the quasicharacteristic function of the pure quantum state $\ket{\chi}\!\!\bra{\chi}$ and where $\Char{\chi}\bar\otimes \Char{\chi}:\Real^{2\dm}\times\Real^{2\dm}\to\Complex$
    is defined by 
    \begin{equation}
    (\Char{\chi}\bar\otimes \Char{\chi})(\xi,\omega) := \Char{\chi}(\xi)\Char{\chi}(-\omega).
    \end{equation}
\end{lemma}
\begin{proof}
    We have:
    \begin{equation}
    \label{reproducing-computation}
    \begin{split}
    \Matel{\A}(\alpha,\beta) &=
    \braket{\chi_\alpha|\A|\chi_\beta} \\
    &=
    \frac{1}{(2\pi)^{\dm}}
    \int \braket{\chi_\alpha|\chi_{\alpha'}}
    \braket{\chi_{\alpha'}|\rho|\chi_\beta} \diff\alpha'\\
    &=
    \frac{1}{(2\pi)^{2\dm}}
    \int \braket{\chi_\alpha|\chi_{\alpha'}}
    \braket{\chi_{\alpha'}|\rho|\chi_{\beta'}}
    \braket{\chi_{\beta'}|\chi_\beta} \diff\alpha'\diff\beta'\\
    &= 
    \frac{1}{(2\pi)^{2\dm}}
    \int e^{i\alpha'\wedge\alpha/2} \Char{\chi}(\alpha'-\alpha)
    \Matel{\A}(\alpha',\beta')\Char{\chi}(\beta-\beta')e^{i\beta\wedge\beta'/2}\diff\alpha'\diff\beta',\\
    &= 
    \frac{1}{(2\pi)^{2\dm}}
    ((\Char{\chi}\bar\otimes \Char{\chi})\tconv_{\Omega'} \Matel{\A})(\alpha,\beta)
    \end{split}
    \end{equation}
where to get the second line we use Eq.~\eqref{rank-1-trace} of Lemma~\ref{lem:trace} for the inner product of $\ket{\chi_\alpha}, \rho\ket{\chi_\beta}\in L^2(\Real^\dm)$ and to get the third line we use the lemma again for the inner product of $\ket{\chi_\beta}\in L^2(\Real^\dm), \rho\ket{\chi_\alpha'}\in L^2(\Real^\dm)$. 
The final line is just the definition of the twisted convolution with respect to the form $\Omega'$.
\end{proof}

\begin{corollary}
    \label{cor:matel-schwartz}
    If the matrix element $\Matel{\A}$ is rapidly decaying, then it is a Schwartz function.
\end{corollary}
\begin{proof}
    First note that $\chi$ and therefore $\Kernel{\chi}(x,y)=\chi(x)\bar\chi(y)$ are Schwartz functions. By Lemma~\ref{lem:wigner-kernel-schwartz} this means $\Wigner{\chi}$ is a Schwartz function, so therefore $\Char{\chi}$ is a Schwartz function by Lemma~\ref{lem:wigner-char-relation}, meaning that $\Char{\chi}\bar\otimes \Char{\chi}$ is a Schwartz function.  Then by Lemma~\ref{lem:reproducing}, 
    $\Matel{\A}$ is the twisted convolution of a Schwartz function ($\Char{\chi}\bar\otimes \Char{\chi}$)
    against a rapidly decaying function ($\Matel{\A}$, by assumption), and is therefore also Schwartz-class by Lemma~\ref{lem:twisted-convolution}.
\end{proof}
We now deploy Lemma~\ref{lem:wigner-char-relation} to recover the Wigner function $\Wigner{\A}$ from the matrix element $\Matel{\A}$.

\begin{lemma}
    \label{lem:wigner-from-matel}
    For any quantum state $\A$,
    \begin{equation}
    \label{double-char}
    \Wigner{\A}(\alpha) =
    \frac{1}{(2\pi)^{3\dm}}
    \int  e^{-i(\alpha-\beta/2)\wedge \xi}
    \Matel{\A}(\beta-\xi/2,\beta+\xi/2)
     \diff\xi\diff\beta.
    \end{equation}
\end{lemma}
\begin{proof}
We have:
\begin{equation}
\begin{split}
\Wigner{\A}(\alpha)
&= \frac{1}{(2\pi)^{2\dm}} \int  e^{-i\alpha\wedge\xi} \trace[D_{\xi/2}\A
D_{\xi/2}]  \diff\xi\\
&= \frac{1}{(2\pi)^{3\dm}} \int 
e^{-i\alpha\wedge\xi} \braket{\chi_\beta|D_{\xi/2}\A
D_{\xi/2}|\chi_\beta}  \diff\xi\diff\beta\\
&= \frac{1}{(2\pi)^{3\dm}}\int
e^{-i\alpha\wedge\xi} \braket{\chi|D_{-\beta}D_{\xi/2}\A
D_{\xi/2}D_\beta|\chi} \diff\xi\diff\beta \\
&= \frac{1}{(2\pi)^{3\dm}} \int 
e^{-i(\alpha-\beta/2)\wedge\xi} \braket{\chi|D_{\xi/2-\beta}\A
D_{\xi/2+\beta}|\chi}  \diff\xi\diff\beta\\
&= \frac{1}{(2\pi)^{3\dm}} \int  e^{-i(\alpha-\beta/2)\wedge \xi}
\Matel{\A}(\beta-\xi/2,\beta+\xi/2) \diff\xi\diff\beta,
\end{split}
\end{equation}
where for the first line we use Lemma~\ref{lem:wigner-char-relation}, for the second line we use the trace formula 
from Lemma~\ref{lem:trace}, and for the fourth line we use the composition identity $D_\alpha D_\beta = e^{i\beta\wedge\alpha/2}D_{\alpha+\beta}$ for the displacement operator.
\end{proof}

\begin{corollary}
    \label{cor:matel-wigner-schwartz}
	For any quantum state $\A$, if the matrix element $\Matel{\A}$ is a Schwartz function, then $\Wigner{\A}$ is a Schwartz function.
\end{corollary}
\begin{proof}
    Lemma~\ref{lem:wigner-from-matel} shows that $\Wigner{\A}$ can be obtained from $\Matel{\A}$ by (a) multiplying by the phase function $e^{i\beta\wedge\xi/2}$ (quadratic in the variables $\beta$ and $\xi$), (b) applying a symplectic Fourier transform (exchanging the variable $\xi$ for the variable $\alpha$), and then (c) integrating over the variable $\beta$. All three of these operations preserve Schwartz-class functions, so $\Wigner{\A}$ is a Schwartz function.
\end{proof}
With all the hard work done, our main result follows easily.
\newtheorem*{theorem:main}{Theorem \ref{thm:main}}
\begin{theorem:main}
    \mainTheorem
\end{theorem:main}
\begin{proof}
    The rapid decay of $\Wigner{\rho}$ means $\infty > \int \Wigner{\rho}(\alpha)\diff \alpha=\trace[\rho]$, so $\rho$ is trace-class and hence a quantum state.
    We then conclude that $\Matel{\A}$ is rapidly decaying by Corollary~\ref{cor:rapid-decay}, so $\Matel{\A}$ is a Schwartz function by Corollary~\ref{cor:matel-schwartz}.  Therefore, $\Wigner{\A}$ is a Schwartz function by Corollary~\ref{cor:matel-wigner-schwartz}.
\end{proof}

This proof is constructive and so can in principle be used to derive effective bounds for the Schwartz seminorms for $\Wigner{\A}$ in terms of only the decay seminorms.  However, computing the bounds through this method would be very laborious, so instead we use a more direct method in the next subsection.


\subsection{Direct proof}
\label{sec:direct-proof}

The strategy rests on showing that the Schwartz seminorms
of $\Wigner{\ket{\chi_\alpha}\!\bra{\chi_\beta}}$ depend
only polynomially on $\alpha$ and $\beta$.  In
this section, we will use for convenience the Schwartz
type norms
\begin{equation}
    \|F\|_{a,b}
    := \sum_{a'\leq a}\sum_{b'\leq b}
    |F|_{a',b'}
\end{equation}
and the corresponding shorthand $\|F\|_a = \|F\|_{a,0}=\sum_{a'\leq a}|F|_{a',0}$.

\begin{lemma}
    \label{lem:wigner-decomp}
    For any quantum state $\rho$ and any reference wavefunction $\chi\in\Schwartz(\Real^\dm)$,
    \begin{equation}
    \label{eq:wigner-decomp}
        \Wigner{\rho}(\gamma) = \frac{1}{(2\pi)^{2\dm}}
        \int \Wigner{\ket{\chi_\alpha}\!\bra{\chi_\beta}} (\gamma)
        \Matel{\A}(\alpha,\beta)
        \diff\alpha\diff\beta.
    \end{equation}
\end{lemma}
\begin{proof}
    Given the spectral decomposition $\rho = \sum_j \lambda_j\ket{\psi_j}\!\!\bra{\psi_j}$ we use Lemma~\ref{lem:wigner-char-relation} to get
    \begin{equation}
    \begin{split}
        \Wigner{\A}(\gamma) &=
        \frac{1}{(2\pi)^{2\dm}}
        \int e^{-i\gamma\wedge\xi}\trace[\A D_\xi]\diff \xi\\
        &=\frac{1}{(2\pi)^{4\dm}}
        \int e^{-i\gamma\wedge\xi}\braket{\chi_\beta|D_\xi|\chi_\alpha}
        \braket{\chi_\alpha|\A|\chi_\beta}\diff\alpha\diff\beta\diff\xi\\
        &= \frac{1}{(2\pi)^{4\dm}}
        \int e^{-i\gamma\wedge\xi}\Char{\ket{\chi_\alpha}\!\bra{\chi_\beta}}(\xi)
        \Matel{\A}(\alpha,\beta)\diff\alpha\diff\beta\diff\xi
    \end{split}
    \end{equation}
    where to get the second line we apply Lemma~\ref{lem:trace} to the trace-class operator $\A D_\xi$ twice. Using Lemma~\ref{lem:wigner-char-relation} yields \eqref{eq:wigner-decomp}.
\end{proof}

\begin{lemma}
    \label{lem:off-diagonal-wigner}
    The Schwartz seminorms of the Wigner transform of the off-diagonal operator $\ket{\chi_\alpha}\!\!\bra{\chi_\beta}$ obey
    \begin{align}
        \label{Wab-schwartz-bound-tight}
        |\Wigner{\ket{\chi_\alpha}\!\bra{\chi_\beta}}|_{a,b}
        &\le \sum_{c\le b}\sum_{d\le a}\sum_{e\le c}\sum_{f\le d} \binom{b}{c}\binom{a}{d}\binom{c}{e}\binom{d}{f}2^{-|d|}
        |\Wigner{\chi}|_{a-d,b-c}
        |\alpha^{\hat{e}+f}\beta^{\hat{c}-\hat{e}+d-f}| \\
        \label{Wab-schwartz-bound-loose}
        &\leq 4^{|a|+|b|}(1 + |\alpha| + |\beta|)^{|a|+|b|} \|\Wigner{\chi}\|_{a,b},
    \end{align}
    where we use a hat to swap the position and momentum components of a multi-index: $\hat{a}=\widehat{(a_\rmx,a_\rmp)} :=(a_\rmp,a_\rmx)$.  
\end{lemma}
\begin{proof}
    First note that 
    \begin{equation}
    \begin{split}
        \Char{\ket{\chi_\alpha}\!\bra{\chi_\beta}}(\xi) &= \trace[\ket{\chi_\alpha}\!\!\bra{\chi_\beta} D_\xi]
        = \braket{\chi|D_{-\beta}D_\xi D_\alpha|\chi}
        \\
        &= e^{i(\alpha+\beta)\wedge\xi/2+i\beta\wedge\alpha/2}
        \braket{\chi|D_{\xi+\alpha-\beta}|\chi}
        = e^{i\bar\alpha\wedge\xi+i\bar\alpha\wedge\Delta\alpha/2}
        \Char{\chi}(\xi+\Delta\alpha)
    \end{split}
    \end{equation}
    where in the last line we introduced the shorthand $\bar\alpha = (\alpha+\beta)/2$ and $\Delta\alpha = \alpha - \beta$.  Then $\Wigner{\ket{\chi_\alpha}\!\bra{\chi_\beta}}$ is related to $\Wigner{\chi}$ by
    \begin{equation}
    \begin{split}
        \label{eq:off-diagonal-wigner}
        \Wigner{\ket{\chi_\alpha}\!\bra{\chi_\beta}}(\gamma) &= 
        \int e^{-i\gamma\wedge\xi} \Char{\ket{\chi_\alpha}\!\bra{\chi_\beta}}(\xi) \diff\xi
        \\
        &= 
        \int e^{i(\bar\alpha-\gamma)\wedge\xi+i\bar\alpha\wedge\Delta\alpha/2}
        \Char{\chi}(\xi+\Delta\alpha) \diff\xi
        \\
        &= 
        \int e^{i(\bar\alpha-\gamma)\wedge(\xi-\Delta\alpha)+i\bar\alpha\wedge\Delta\alpha/2}
        \Char{\chi}(\xi) \diff\xi\\
        &= 
        e^{i(\gamma-\bar\alpha/2)\wedge\Delta\alpha} \Wigner{\chi}(\gamma-\bar\alpha),
    \end{split}
    \end{equation}
    so
    \begin{equation}
    \begin{split}
        \gamma^a\partial_\gamma^b
        \Wigner{\ket{\chi_\alpha}\!\bra{\chi_\beta}}(\gamma)
        &= 
        \gamma^a\partial_\gamma^b
        e^{i(\gamma-\bar\alpha/2) \wedge\Delta\alpha}
        \Wigner{\chi} (\gamma-\bar\alpha)\\
        &= \sum_{c\le b} \binom{b}{c}
        e^{i(\gamma-\bar\alpha/2) \wedge\Delta\alpha}
        (i\Omega\cdot \Delta\alpha)^c ((\gamma-\bar\alpha)+\bar\alpha)^a
        \partial_\gamma^{b-c}\Wigner{\chi} (\gamma-\bar\alpha)
    \end{split}
    \end{equation}
    where in the second line we used 
    \begin{equation}
    \begin{split}
    \partial_\gamma^b e^{i\gamma\wedge\xi}= \partial_{\gamma_\rmx}^{b_\rmx}\partial_{\gamma_\rmp}^{b_\rmp}e^{i\gamma_\rmx \cdot\xi_\rmp-i\gamma_\rmp \cdot\xi_\rmx} = (i\xi_\rmp)^{b_\rmx}(-i\xi_\rmx)^{b_\rmp} e^{i\gamma_\rmx \cdot\xi_\rmp-i\gamma_\rmp \cdot\xi_\rmx} = (i\Omega\cdot\xi)^b e^{i\gamma\wedge\xi}
    \end{split}
    \end{equation}
    Then, using $|(i\Omega\cdot\xi)^b|=|\xi_\rmp^{b_\rmx}\xi_\rmx^{b_\rmp}| = |\xi^{\hat b}|$, it follows that
    \begin{equation}
    \begin{split}
        |\Wigner{\ket{\chi_\alpha}\!\bra{\chi_\beta}}|_{a,b}
        &\leq \sum_{c\le b}\sum_{d\le a} \binom{b}{c}\binom{a}{d} 
        |\Wigner{\chi}|_{a-d,b-c}
        |\Delta\alpha^{\hat{c}}\bar\alpha^d|\\
        &\le \sum_{c\le b}\sum_{d\le a}\sum_{e\le c}\sum_{f\le d} \binom{b}{c}\binom{a}{d}\binom{c}{e}\binom{d}{f}2^{-|d|}
        |\Wigner{\chi}|_{a-d,b-c}
        |\alpha^{\hat{e}+f}\beta^{\hat{c}-\hat{e}+d-f}|\\
        &\leq 
        4^{|a|+|b|}(1+|\alpha|+|\beta|)^{|a|+|b|}
        \|\Wigner{\chi}\|_{a,b}.
    \end{split}
    \end{equation}
    To get the third line, we bound the terms in the sum on the second line using $2^{-|d|}\le 1$, $|\Wigner{\chi}|_{a-d,b-c} \le \|\Wigner{\chi}\|_{a,b}$, and $|\alpha^{\hat{e}+f}\beta^{\hat{c}-\hat{e}+d-f}| \le (|1+|\alpha|+|\beta|)^{|a|+|b|}$ and then sum the binomial coefficients.
\end{proof}

\begin{theorem}
\label{thm:seminorm-bound}
For any quantum state $\A$ and reference state $\chi\in\Schwartz(\Real^\dm)$, the following inequality holds:  
\begin{align}
    |\Wigner{\A}|_{a,b} 
    \label{full-decay-bd}
    &\le (2\pi)^{5\dm} 2^{4(|a|+|b|+\dm)}
    \|\Wigner{\chi}\|_{a,b}\|\Wigner{\chi}\|_{2(a+\hat{b})+6}
    \|\Wigner{\A}\|_{2(a+\hat{b})+4}.
\end{align}
\end{theorem}
Note that the right-hand side contains only the decay 
norms of $\Wigner{\rho}$, and the left-hand side contains
an arbitrary Schwartz seminorm, so this implies Theorem~\ref{thm:main}.  We also observe that, on the right-hand side, the position and momentum indices are flipped in the derivative multi-index $\hat{b} = \widehat{(b_\rmx,b_\rmp)} = (b_\rmp,b_\rmx)$ when it contributes to a coordinate power (rather than a derivative power).
\begin{proof}
We start with Eq.~\eqref{eq:wigner-decomp} of Lemma~\ref{lem:wigner-decomp} and apply Lemma~\ref{lem:husimi-bound-matel} followed by Eq.~\eqref{Wab-schwartz-bound-tight} from Lemma~\ref{lem:off-diagonal-wigner} to get
\begin{equation}
\begin{split}
    |\gamma^a\partial_\gamma^a\Wigner{\A}(\gamma)| 
    &\le \frac{1}{(2\pi)^{2\dm}}
    \int |\gamma^a\partial_\gamma^a\Wigner{\ket{\chi_\alpha}\!\bra{\chi_\beta}} (\gamma)| 
    |\Husimi{\A}(\alpha)|^{1/2} 
    |\Husimi{\A}(\beta)|^{1/2}
    \diff\alpha\diff\beta\\
    &\le \frac{1}{(2\pi)^{2\dm}}\sum_{c\le b}\sum_{d\le a}\sum_{e\le c}\sum_{f\le d} \binom{b}{c}\binom{a}{d}\binom{c}{e}\binom{d}{f}2^{-|d|}
    |\Wigner{\chi}|_{d-a,b-c}\\
    &\qquad \times \int  \frac{1}{\prod_{j=1}^{2\dm}(1+|\alpha_j|^2) }
    |\alpha^{\hat{e}+f}|
    \prod_{j=1}^{2\dm}(1+|\alpha_j|^2) |\Husimi{\A}(\alpha)|^{1/2}\diff\alpha\\
    &\qquad \times \int  \frac{1}{\prod_{j=1}^{2\dm}(1+|\beta_j|^2) }
    |\beta^{\hat{c}-\hat{e}+d-f}|
    \prod_{j=1}^{2\dm}(1+|\beta_j|^2)|\Husimi{\A}(\beta)|^{1/2}\diff\beta.
\end{split}
\end{equation}
To compute the integral over $\alpha$, we use
\begin{equation}
\begin{split}
    |\alpha^{\hat{e}+f}| \prod_{j=1}^{2\dm}(1+|\alpha_j^2|)|\Husimi{\A}(\alpha)|^{1/2}
    &\leq 2^{2\dm}\|\Husimi{\A}\|_{2\hat{e}+2f+4}^{1/2} 
    \leq 2^{2\dm} \|\Husimi{\A}\|_{2(a+\hat{b})+4}^{1/2},
\end{split}
\end{equation}
and likewise for the integral over $\beta$.  Integrating over $\alpha$ with $\int (1+|\alpha_j|^2)^{-1}\diff\alpha_j = \pi$ and likewise for $\beta$, we obtain
\begin{align}
    |\Wigner{\A}|_{a,b} 
    &\le \left(\frac{\pi}{2}\right)^{2\dm}
    \|\Husimi{\A}\|_{2(a+\hat{b})+4} 
    \sum_{c\le b}\sum_{d\le a}\sum_{e\le c}\sum_{f\le d} \binom{b}{c}\binom{a}{d}\binom{c}{e}\binom{d}{f}2^{-|d|}
    |\Wigner{\chi}|_{d-a,b-c} \nonumber \\
    \label{Husimi-decay-bd}
    &\le \left(\frac{\pi}{2}\right)^{2\dm}
    2^{2(|a|+|b|)}\|\Wigner{\chi}\|_{a,b}\|\Husimi{\A}\|_{2(a+\hat{b})+4}
\end{align}
using $2^{-|d|}\le 1$ and $|\Wigner{\chi}|_{d-a,b-c}\le\|\Wigner{\chi}\|_{a,b}$.
We then bound the decay seminorms of $\Husimi{\A}$ using Lemma~\ref{lem:husimi-convolution}:
\begin{equation}
\begin{split}
    \alpha^a\Husimi{\A}(\alpha) = 
    (2\pi)^\dm \int \Wigner{\A}(\alpha-\beta) ((\alpha-\beta)+\beta)^a \WignerMinus{\chi}(\beta)
    \diff\beta
\end{split}
\end{equation}
so 
\begin{equation}
\begin{split}
    |\alpha^a\Husimi{\A}(\alpha)| &\le 
    (2\pi)^\dm \sum_{b\le a} \binom{a}{b}\int \frac{1}{\prod_{j=1}^{2\dm}(1+|\beta_j|^2)} |(\alpha-\beta)^b \Wigner{\A}(\alpha-\beta)|\\
    &\qquad\qquad\qquad\qquad\qquad \times |\beta^{a-b}|\prod_{j=1}^{2m}(1+|\beta_j|^2)|\WignerMinus{\chi}(\beta)|\diff \beta \\
    &\leq (2\pi)^n 2^{|a|} \pi^{2n}
    \|\Wigner{\A}\|_a \|\Wigner{\chi}\|_{a+2}
\end{split}
\end{equation}
using $\|\Wigner{\chi}\|_{a+2} = \|\WignerMinus{\chi}\|_{a+2}$. 
Summing over the seminorms in the norm,
\begin{align}
    \|\Husimi{\A}\|_a = \sum_{b\le a} |\Husimi{\A}|_b
    &\le 2^n \pi^{3n}  \sum_{b\le a} 2^{|b|} 
    \|\Wigner{\A}\|_b \|\Wigner{\chi}\|_{b+2} \\
    &\le 2^{3n+|a|}\pi^{3n} \|\Wigner{\A}\|_a \|\Wigner{\chi}\|_{a+2},
\end{align}
and then inserting into \eqref{Husimi-decay-bd} yields \eqref{full-decay-bd}.
\end{proof}





\section{Schwartz states}
\label{sec:schwartz-states}



In this section, we extend our main result by establishing an equivalence between the Schwartz-class and rapid-decay properties of many different representations of the quantum state.  
First, we will give a notion of Schwartz class to a set of orthogonal wavefunction $\{\psi_i\}$ appearing in a spectral decomposition $\Kernel{\A}(x,y)=\sum_i\psi_i(x)\bar{\psi}_i(y)$.
Then, we recall the definition of a Schwartz operator as identified by Keyl et al.~\cite{keyl2016schwartz}.  Finally, we will prove our large equivalence theorem.


To guarantee that $\Kernel{\A}$ is a Schwartz function, it is, of course, not sufficient for each $\psi_i$ to be a Schwartz function.  For example, if each $\psi_i$ is a Gaussian wavepacket with increasing variance $\sigma^2_i \propto i$ centered on the origin, then the overall variance $\langle X^2 \rangle = \Tr[X^2 \A] = \int x^2 \Kernel{\A}(x,x) \diff x = \sum_i\int x^2 |\psi_i(x)|^2\diff x = \sum_i p_i \sigma_i^2$ can diverge if the norms $p_i = \braket{\psi_i|\psi_i}$ are decreasing slowly, so that $\Kernel{\A}$ is not a Schwartz function even though each $\psi_i$ is.  Instead, we consider the following definition.\footnote{Note that the multi-indices $a,b,c,d$ in this section are $\dm$-dimensional rather than $2\dm$-dimensional because the wavefunction $\psi$ and the kernel $\Kernel{\A}$ take arguments in position space rather than phase space.}
\begin{definition}
    \label{def:jointly-schwartz}
     A set $\{\psi_i\}$ of unnormalized wavefunctions ($\psi_i\in L^2(\Real^{\dm})$ for all $i$) is \emph{jointly Schwartz} when $|\{\psi_i\}|_{a,b} < \infty$ for all $a,b \in (\bbN \cup \{0\})^{\times \dm}$, where the Schwartz seminorms of such a set are defined by
    \begin{equation}
        \label{jointly-schwartz}
        |\{\psi_i\}|_{a,b}^2 := \sup_{x}\sum_i \left|x^a\partial_x^b\psi_i(x)\right|^2.
    \end{equation}
    This is denoted $\{\psi_i\}\in \Schwartz_{\mathrm{j}}(\Real^{\dm})$.
\end{definition}
Note that this seminorm \eqref{jointly-schwartz} is not simply a function of the seminorm $|\Kernel{\A}|_{(a,c),(b,d)}$ of the kernel, nor is it simply a function of the individual seminorms $|\psi_i|_{a,b}:=\sup_x |x^a\partial_x^b\psi_i(x)|$ of the wavefunctions $\psi_i$.
However, one can check that when the set $\{\psi_i\}$ is finite, the jointly Schwartz property is equivalent to the condition that all the wavefunctions are Schwartz functions individually, $\{\psi_i\}\subset\Schwartz(\Real^{\dm})$.

\begin{lemma}
\label{lem:jointly-schwartz}
For any quantum state $\A$, the set of unnormalized wavefunctions $\{\psi_i\}$ of the spectral decomposition is jointly Schwartz if and only if the kernel $\Kernel{\A}$ is a Schwartz function.
\end{lemma}
\begin{proof}
First assume that $\{\psi_i\}\in \Schwartz_{\mathrm{j}}(\Real^{\dm})$.  Then
\begin{equation}
\begin{split}
|\Kernel{\A}|_{(a,c),(b,d)} &= \sup_{x,y}\left|x^a y^c\partial_x^b\partial_y^d K(x,y)\right|\\
&=\sup_{x,y} \left|\sum_i \left(x^a\partial_x^b\psi_i(x)\right)\left( y^{c}\partial_y^{d}\bar\psi_i(y)\right)\right|\\
&\le \sup_{x,y} \left(\sum_i \left|x^a\partial_x^b\psi_i(x)\right|^2\right)^{1/2}\left(\sum_i \left|y^{c}\partial_y^{d}\psi_i(y)\right|^2\right)^{1/2}\\
&=  \left(\sup_{x}\sum_i \left|x^a\partial_x^b\psi_i(x)\right|^2\right)^{1/2}\left(\sup_{y}\sum_i \left|y^{c}\partial_y^{d}\psi_i(y)\right|^2\right)^{1/2}\\
&=|\{\psi_i\}|_{a,b}|\{\psi_i\}|_{c,d},
\end{split}
\end{equation}
where the third line is the Cauchy-Schwartz inequality.
Therefore, $\{\psi_i\}\in \Schwartz_{\mathrm{j}}(\Real^{\dm})\Rightarrow\Kernel{\A}\in \Schwartz(\Real^{2\dm})$. To see the inverse, note that 
\begin{equation}
\begin{split}
|\Kernel{\A}|_{(a,a),(b,b)} &= \sup_{x,y}\left|x^a y^a\partial_x^b\partial_y^b K(x,y)\right|\\
&\ge \sup_{x}\left|\left(x^a y^a\partial_x^b\partial_y^b K\right)(x,x)\right|\\
&=\sup_x \sum_i |x^a\partial_x^b\psi_i|^2\\
\end{split}
\end{equation}
where the inequality holds because $\Kernel{\A}$ is a Schwartz function.  This quantity diverges by definition if $\{\psi_i\}\notin \Schwartz_{\mathrm{j}}(\Real^{\dm})$, implying $\Kernel{\A}\notin \Schwartz(\Real^{2\dm})$.
\end{proof}

The natural way to characterize the Schwartz class of quantum states was suggested by Keyl et al.~\cite{keyl2016schwartz} as those quantum states $\rho$ with bounded expectation value for all symmetric polynomials in $X$ and $P$ (i.e., $\trace[X^a P^b\rho P^b X^a]<\infty$ for all $a,b\in(\mathbb{N}\cup\{0\})^{\times \dm}$).  More generally, for arbitrary operators (i.e., not necessarily positive semidefinite, self-adjoint, or trace-class), they define:
\begin{definition}
    \label{def:schwartz-operator}
    An operator $E$ is a Schwartz operator, denoted $E\in \Schwartz(L^2(\Real^{\dm}))$, when $|E|_{a,b,c,d} < \infty$ for all $a,b,c,d \in (\bbN \cup \{0\})^{\times\dm}$, where the Schwartz operator seminorms are
    \begin{equation}
    |E|_{a,b,c,d} := \sup_{|\psi|,|\phi| = 1 } \left|\left\langle \psi \left|X^a P^b E P^c X^d \right| \phi \right\rangle\right|.
    \end{equation}
    Here, the supremum is taken over all normalized wavefunctions $\psi,\phi \in L^2(\Real^{\dm})$.
\end{definition}

Equipped with Definitions~\ref{def:jointly-schwartz} and \ref{def:schwartz-operator}, we can state a theorem that subsumes the results of Sec. ~\ref{sec:proof}.

\begin{theorem}
\label{thm:mega-equivalence}
For any Schwartz-class reference wavefunction $\chi\in\Schwartz(\Real^\dm)$ and for any quantum state 
(i.e., positive semidefinite trace-class operator on $L^2(\Real^\dm)$)
$\A$, with spectral decomposition $\{\psi_i\}$, quasicharacteristic function $\Char{\A}$, Wigner function $\Wigner{\A}$, kernel $\Kernel{\A}$, Husimi function $\Husimi{\A}$, and matrix element $\Matel{\A}$, the following conditions are equivalent:
\begin{itemize}
	\item $\Wigner{\A}\in \Schwartz(\Real^{2\dm})$
	\item $\Wigner{\A}\in \Decay(\Real^{2\dm})$
	\item $\Husimi{\A}\in \Schwartz(\Real^{2\dm})$
	\item $\Husimi{\A}\in \Decay(\Real^{2\dm})$
	\item $\Matel{\A}\in \Schwartz(\Real^{4d})$
	\item $\Matel{\A}\in \Decay(\Real^{4d})$
	\item $\Char{\A}\in \Schwartz(\Real^{2\dm})$
    \item $\Kernel{\A}\in \Schwartz(\Real^{2\dm})$
    \item $\{\psi_i\}\in \Schwartz_{\mathrm{j}}(\Real^{\dm})$
    \item $\A \in \Schwartz(L^2(\Real^{\dm}))$ 
\end{itemize}
Furthermore, if the set $\{\psi_i\}$ is finite (e.g., if the state is pure, $\A=|\psi\rangle\!\langle\psi|$), then the condition $\{\psi_i\}\subset\Schwartz(\Real^{\dm})$ is also equivalent to the above.
\end{theorem}
\begin{proof} We have:
        
    $\Wigner{\A}\in\Schwartz(\Real^{2\dm})\Rightarrow\Wigner{\A}\in\Decay(\Real^{2\dm})$ because $\Schwartz(\Real^{2\dm})\subset\Decay(\Real^{2\dm})$.

    $\Wigner{\A}\in\Decay(\Real^{2\dm}) \Rightarrow \Husimi{\A}, \Matel{\A}\in\Decay(\Real^{2\dm})$ by 
    Corollary~\ref{cor:rapid-decay}.
    
    $\Matel{\A}\in\Decay(\Real^{2\dm})\Rightarrow \Matel{\A}\in\Schwartz(\Real^{2\dm})$ by Corollary~\ref{cor:matel-schwartz}.
    
    $\Matel{\A}\in\Schwartz(\Real^{2\dm})\Rightarrow \Wigner{\A}\in\Schwartz(\Real^{2\dm})$ by Corollary~\ref{cor:matel-wigner-schwartz}.
    
    $\Husimi{\A}\in\Schwartz(\Real^{2\dm})\Rightarrow\Husimi{\A}\in\Decay(\Real^{2\dm})$ because $\Schwartz(\Real^{2\dm})\subset\Decay(\Real^{2\dm})$.
    
    $\Matel{\A}\in\Schwartz(\Real^{2\dm})\Rightarrow\Matel{\A}\in\Decay(\Real^{2\dm})$ because $\Schwartz(\Real^{2\dm})\subset\Decay(\Real^{2\dm})$.

    
    
    
    
    $\Wigner{\A}\in\Schwartz(\Real^{2\dm}) \Leftrightarrow\Char{\A}\in\Schwartz(\Real^{2\dm})$ by Lemma~\ref{lem:wigner-char-relation}.
    
    $\Wigner{\A}\in\Schwartz(\Real^{2\dm}) \Leftrightarrow\Kernel{\A}\in\Schwartz(\Real^{2\dm})$ by Lemma~\ref{lem:wigner-kernel-schwartz}.
    
    $\Kernel{\A}\in \Schwartz(\Real^{2\dm})\Leftrightarrow \{\psi_i\}\in \Schwartz_{\mathrm{j}}(\Real^{2\dm})$ by Lemma~\ref{lem:jointly-schwartz}.
    
    $\A \in \Schwartz(L^2(\Real^{\dm}))\Leftrightarrow \Wigner{\A}\in\Schwartz(\Real^{2\dm})$ by Proposition 3.18 in Ref.~\cite{keyl2016schwartz}. 
\end{proof}
We say a quantum state satisfying the above equivalent conditions is a \emph{Schwartz state}. 

Note that $\Kernel{\A}\in \Decay(\Real^{2\dm})$ is \emph{not} an equivalent condition, being strictly weaker than the other conditions above.\footnote{For instance, the plateau wavefunction $\psi(y) = \{1 \,\mathrm{if}\, 0\le y \le 1, 0\,\mathrm{otherwise}\}$ is compactly supported in position space but in momentum space decays to infinity only as a polynomial.}
This is essentially because $\Kernel{\A}$ is a spatial representation, so momentum information is encoded only in its derivatives, whereas $\Wigner{\A}$, $\Husimi{\A}$, and $\Matel{\A}$ are phase-space representations whose decay constrains both space and momentum features. Similarly, $\Char{\A}\in \Decay(\Real^{2\dm})$ is not an equivalent condition because rapid decay of the derivatives of $\Wigner{\A}$ does not assure that $\Wigner{\A}$ has rapid decay.\footnote{Consider the $\dm=1$ quantum state $\rho=\sum_{k=0}^\infty |\psi_k\rangle\!\langle\psi_k|$ with $\psi_k(y) = \psi_0(y-z_k) = (6/\pi^2)k^{-2} \exp[-(y-z_k)^2/2]/\sqrt{2\pi}$ with $z_k = k^3$.  This is a mixture of Gaussians of equal variance, so the derivatives are all rapidly decreasing and $\Char{\A}\in\Decay(\Real^{2\dm})$, but the mean $\trace[\rho X]$ diverges so $\Wigner{\A}$ is not a Schwartz function.}

\section{Discussion}
\label{sec:discussion}

Although the Wigner formalism provides a complete representation of quantum states and dynamics, it is often regarded as less fundamental.  (It only really becomes uniquely preferred in the classical limit, and under, e.g., certain symmetry demands to distinguish it from other deformations of classical mechanics; see for instance the introduction of Ref.~\cite{degosson2016bornjordan}.)  One practical reason is that computations are often more difficult using the Moyal product\footnote{Of course, the peculiar features of the Moyal product are not just a matter of practicalities: Because the Moyal bracket has phase-space derivatives of arbitrarily high order, the dynamics of the Wigner function are non-local.}, the so-called ``$\star$-genvalue equations'', and so on \cite{curtright2014concise}. Another reason is that the state space is awkward to define. 

Formally, the state space of valid Wigner functions can be delineated with the quantum generalization \cites{srinivas1975nonclassical,broecker1995mixed} of Bochner's theorem \cite{bochner1933monotone}.  (See also illuminating discussion and further generalizations to some discrete spaces in Ref.~\cite{dangniam2015quantum}.) This definition is sufficiently opaque that most physicists simply think of the allowed pure-state Wigner functions as the image of the Wigner transform of the space of allowed quantum states, $L^2(\Real^{\dm})$, if they think of it at all.  In particular, many simple (even positive-valued) functions on $L^1(\Real^{2\dm})$ are not the Wigner functions of any quantum states.

In contrast, the $L^2(\Real^{\dm})$ (pure) state space of the Schr\"odinger representation is relatively simple to understand, and the ``parameterization'' of that space is natural in the sense that all possible functions are allowed modulo only the single, easy-to-interpret constraint of normalization. This becomes even clearer in the case of a finite-dimensional quantum system, where there are no complications related to the continuum and where any complex-valued function over configuration space suffices as a (not necessarily normalized) state.  Delineating the corresponding set of Wigner functions for finite-dimensional systems is much more subtle \cite{dangniam2015quantum}.
	
In this sense, the Wigner representation is ``overparameterized''.  One can think of our Theorem~\ref{thm:main} as better characterizing this overparameterization: the regularity of the \emph{interior} of the Wigner function in terms of its derivatives of any order is tightly controlled by the Wigner function's decay toward infinity, a feature that is obviously not shared by all normalized functions over $\Real^{2\dm}$. With the seminorm bound of Theorem~\ref{thm:seminorm-bound}, one can also recover a version of the uncertainty principle.  For example, because the supremum of the gradient of $\Wigner{\rho}$ is bounded by the decay seminorms of $\Wigner{\rho}$, it is impossible for $\Wigner{\rho}$ to be supported in a ball of too small a radius.


\section{Acknowledgements}

CJR thanks Jukka Kiukas and Reinhard Werner for helpful discussion.  FH was supported by the Fannie and John Hertz Foundation Fellowship.



\appendix
\section{}

Here we here recall the proofs of some standard results referenced in the main body of this paper.

\newtheorem*{lemma:displacement}{Lemma \ref{lem:displacement}}
\begin{lemma:displacement}
\displacementlemma
\end{lemma:displacement}
\begin{proof}
First, let us prove the statement for some Schwartz function $\chi\in\Schwartz(\Real^{\dm})$. Consider the following differential equation:
\begin{equation}
\label{D-pde}
\partial_t f_t = i R \wedge\xi f_t
\end{equation}
Solutions $f_t$ preserve the $L^2$ norm (because the operator on the
right is anti-Hermitian), so solutions are unique.
Moreover, since $D_{t\xi}f_0$ is a solution to~\eqref{D-pde}, any
solution to~\eqref{D-pde} with initial condition $f_0=\chi$
satisfies $f_1=D_\xi\chi$.

It remains to check that the function
\begin{equation}
f_t(y) = e^{i(y-t\xi_\rmx/2)\cdot (t\xi_\rmp)} \chi(y-t\xi_\rmx)
\end{equation}
solves~\eqref{D-pde}.  We check this by a direct computation,
\begin{equation}
\begin{split}
%
\partial_t f_t(y) &= i((y-t\xi_\rmx)\cdot \xi_\rmp)f_t(y)
- e^{i(y-t\xi_\rmx/2)\cdot (t\xi_\rmp)} \xi_\rmx\cdot (\nabla \chi)(y-t\xi_\rmx) \\
&= i\xi_\rmp\cdot y f_t(y)
- i t\xi_\rmx\cdot \xi_\rmp f_t(y)
- e^{i(y-t\xi_\rmx/2)\cdot (t\xi_\rmp)} \xi_\rmx\cdot (\nabla \chi)(y-t\xi_\rmx) \\
&= i\xi_\rmp\cdot X f_t(y) - i\xi_\rmx\cdot P f_t(y).
\end{split}
\end{equation}
Having proven the claim for $\chi\in\Schwartz(\Real^{\dm})$, we can extend it to any $\phi\in L^2(\Real^{\dm})$ by using the density of the Schwartz function in $L^2(\Real^{\dm})$, i.e., by approximating $D_\xi \phi$ to accuracy $\epsilon$ with some choice of $D_\xi \chi^{(\epsilon)}\in\Schwartz(\Real^{\dm})$ and taking $\epsilon\to 0$.
\end{proof}

\newtheorem*{lemma:wigner-char-relation}{Lemma \ref{lem:wigner-char-relation}}
\begin{lemma:wigner-char-relation}
    For any trace-class kernel operator $E$, the corresponding Wigner transform and quasicharacteristic transform are symplectic Fourier duals:
    \begin{equation}
    \label{wigner-char-relation-appendix}
     \Wigner{E}(\alpha) = \frac{1}{(2\pi)^{2\dm}}\int e^{-i\alpha\wedge\xi}\Char{E}(\xi)\diff\xi.
    \end{equation}
\end{lemma:wigner-char-relation}
\begin{proof}
By linearity it suffices to check the case that $E$ is a rank-1 state
of the form $E=\ket{\psi}\!\!\bra{\phi}$.  Moreover, by the continuity
of the Wigner transform in $L^2(\Real^\dm)\times L^2(\Real^\dm)$ and the density of
Schwartz functions in $L^2(\Real^\dm)$, we may furthermore assume that
$\psi,\phi\in\Schwartz(\Real^\dm)$.
In this case,
\begin{equation}
\begin{split}
\trace[E D_\xi]
&= \braket{\phi|D_\xi|\psi} =
\braket{\phi|D_{\xi/2}D_{\xi/2}|\psi} =
\langle D_{-\xi/2}\phi, D_{\xi/2}\psi\rangle \\
&=\int e^{i(z+\xi_\rmx/4)\cdot \xi_\rmp/2}
e^{i(z-\xi_\rmx/4)\cdot \xi_\rmp/2}
\bar{\phi}(z+\xi_\rmx/2) \psi(z-\xi_\rmx/2)\diff z \\
&= \int e^{iz\cdot \xi_\rmp} \bar{\phi}(z+\xi_\rmx/2)\psi(z-\xi_\rmx/2)\diff z
\end{split}
\end{equation}
where to get the second line we use Lemma~\ref{lem:displacement}.
Applying this identity into the right-hand side of~\eqref{wigner-char-relation-appendix},
we get
\begin{equation}
\begin{split}
\frac{1}{(2\pi)^{2\dm}} \int e^{-i\alpha\wedge\xi}\Char{E}(\xi)\diff\xi &=
\frac{1}{(2\pi)^{2\dm}}\int e^{-i\alpha_\rmx \cdot\xi_\rmp+i\alpha_\rmp \cdot\xi_\rmx} e^{iz\cdot \xi_\rmp}
\bar{\phi}(z+\xi_\rmx/2)\psi(z-\xi_\rmx/2)\diff z \diff \xi_\rmx \diff \xi_\rmp \\
&= \frac{1}{(2\pi)^\dm}\int e^{i \alpha_\rmp\cdot \xi_\rmx}
\bar{\phi}(\alpha_\rmx+\xi_\rmx/2)\psi(\alpha_\rmx-\xi_\rmx/2)\diff \xi_\rmx,
\end{split}
\end{equation}
where the last line follows from the Fourier inversion
formula $\int e^{i(z-\alpha_\rmx)\cdot \xi_\rmp}
f(\alpha_\rmx)\diff z \diff \xi_\rmp = (2\pi)^{\dm} f(z)$. 
This is the definition of the Wigner transform $\Wigner{E}(\alpha)$ of $E=\ket{\psi}\!\!\bra{\phi}$.
\end{proof}

\newtheorem*{lemma:wigner-kernel-schwartz}{Lemma \ref{lem:wigner-kernel-schwartz}}
\begin{lemma:wigner-kernel-schwartz}
    \wignerKernelSchwartzLemma
\end{lemma:wigner-kernel-schwartz}
\begin{proof}
	First assume the kernel $\Kernel{\A}(x,y)$ is a Schwartz function.  Then the function $g(z,\Delta z) = \Kernel{\A}(z-\Delta z/2,z+\Delta z/2)$ is a Schwartz function since it is related to $\Kernel{\A}$ merely by a linear change in variables (a $45^\circ$ rotation).  And of course we have $\Wigner{\A}(x,p)=(2\pi)^{-\dm}\int e^{ip\cdot\Delta z}g(x,\Delta z)\diff \Delta z$, so $\Wigner{\A}$ must also be a Schwartz function since it is just the $\dm$-dimensional Fourier transform of $g$ (exchanging the variable $\Delta z$ for $p$ but leaving the variable $z$).  The argument works the same in the opposite direction, so we conclude that $\Kernel{\A}$ is a Schwartz function if and only if $\Wigner{\A}$ is a Schwartz function.
\end{proof}


\newtheorem*{lemma:twisted-convolution}{Lemma \ref{lem:twisted-convolution}}
\begin{lemma:twisted-convolution}
	\twistedConvolutionLemma
\end{lemma:twisted-convolution}
\begin{proof}
Let $F\in\Schwartz(\Real^{2\dm})$ and $G\in\Decay(\Real^{2\dm})$.  Then recall the 
definition of the twisted convolution,
\begin{equation}
    F\tconv_{\Omega'} G(\alpha) = \int e^{i\alpha\cdot\Omega'\cdot\alpha'/2}F(\alpha-\alpha') G(\alpha')\diff\alpha'.
\end{equation}
We claim that for any multi-index $a=(a_1,\dots,a_{2\dm})\in (\mathbb{N}\cup\{0\})^{\times 2\dm}$,
there exist constants $C(a,a')$ such that 
\begin{equation}
\label{grad-tconv-eq}
\partial^a (F\tconv_{\Omega'} G)(\alpha)
= \sum_{a'\leq a} C(a,a') \int e^{i \alpha\cdot \Omega'\cdot\alpha'/2}(\partial^{a'}F)(\alpha-\alpha') 
(i\Omega'\cdot \alpha')^{a-a'}G(\alpha')\diff\alpha'.
\end{equation}
Equation~\eqref{grad-tconv-eq} is easily checked by induction on $|a|$.   Now applying the 
triangle inequality we can estimate
\begin{equation}
    |\partial^a(F\tconv_{\Omega'} G)(\alpha)|
    \leq \sum_{a'\leq a}
    C(\alpha,\alpha')
    \int |\partial^{a'}F(\alpha-\alpha')|
    |(\Omega'\cdot \alpha')^{a-a'}G(\alpha')|\diff\alpha',
\end{equation}
which is the finite sum of convolutions of the rapidly decaying functions 
$|\partial_\alpha^{a'}F(\alpha)|$ and $|(\Omega'\cdot \alpha)^{a-a'}G(\alpha)|$.  Therefore $\partial^a (F\tconv G)$ is
rapidly decaying.  Since every partial derivative of $F\tconv G$ is rapidly decaying,
it follows that $F\tconv G\in\Schwartz(\Real^{2\dm})$.
\end{proof}
(Note that a similar statement for normal convolutions can be proven in almost exactly the same way.)

\newtheorem*{lemma:hilbert-schmidt-wigner}{Lemma \ref{lem:hilbert-schmidt-wigner}}
\begin{lemma:hilbert-schmidt-wigner}
    \hilbertschmidtwignerlemma
\end{lemma:hilbert-schmidt-wigner}
\begin{proof}
We first decompose $\A$ and $\B$ using the spectral theorem,
\begin{equation}
    \A = \sum_j \lambda_j \ket{\psi_j}\!\!\bra{\psi_j}, \qquad \qquad 
    \B = \sum_j \mu_j \ket{\phi_j}\!\!\bra{\phi_j},
\end{equation}
with $\lambda_j,\mu_j\geq 0$ and $\sum_j \lambda_j = \sum_j \mu_j=1$.
Then $\trace[\A\B] = \sum_{j,k} \lambda_j\mu_k |\!\braket{\psi_j|\phi_k}\!|^2$.  The 
result then follows from
\begin{equation}
\begin{split}
    \int \Wigner{\psi}(\alpha)\Wigner{\phi}(\alpha)\diff\alpha 
    &= \frac{1}{(2\pi)^{2\dm}} 
    \int e^{ip\cdot(y+z)}
    \Kernel{\A}(x-y/2,x+y/2)\\
    &\qquad\qquad\qquad \times
    \Kernel{\B}(x-z/2,x+z/2)
     \diff y \diff z\diff x \diff p\\
     &= \frac{1}{(2\pi)^{\dm}} \int
    \psi(x-y/2)\bar{\psi}(x+y/2)
    \phi(x+y/2)\bar{\phi}(x-y/2)
     \diff y  \diff x \\
     &= \frac{1}{(2\pi)^{\dm}} \int
    \psi(x_-)\bar{\psi}(x_+)
    \phi(x_+)\bar{\phi}(x_-)
     \diff x_+  \diff x_- \\
    &= \frac{1}{(2\pi)^{\dm}} |\braket{\psi|\phi}|^2,
\end{split}
\end{equation}
along with $\Wigner{\rho} = \sum_j \lambda_j \Wigner{\psi_j}$ and 
$\Wigner{\mu}=\sum_j \mu_j\Wigner{\phi_j}$.  The sums can be interchanged because everything converges
absolutely.
\end{proof}

\newtheorem*{lemma:trace}{Lemma \ref{lem:trace}}
\begin{lemma:trace}
    For any trace-class operator $E$ and any $\chi\in L^2(\Real^\dm)$ satisfying $\|\chi\|_{L^2(\Real^\dm)}=1$, 
    \begin{align}
    \trace[E]=\frac{1}{(2\pi)^\dm}\int \braket{\chi_\alpha|E|\chi_\alpha}\diff \alpha.
    \end{align}
In particular, for any $\phi,\psi\in L^2(\Real^\dm)$
    \begin{align}
    \label{inner-prod-formula}
    \braket{\phi|\psi}=\frac{1}{(2\pi)^\dm}\int \braket{\phi|\chi_\alpha}\braket{\chi_\alpha|\psi}\diff \alpha.
    \end{align}
\end{lemma:trace}
\begin{proof}
We start by proving~\eqref{inner-prod-formula}.  For $\psi,\phi\in L^2(\Real^\dm)$,
\begin{equation}
\label{fourier-inv}
\begin{split}
    \int \braket{\phi|\chi_\alpha}\braket{\chi_\alpha|g}\diff\alpha &= \int \Bigg[\left(\int \bar{\phi}(z) e^{i(z-\alpha_\rmx/2)\cdot \alpha_\rmp}
      \chi(z-\alpha_\rmx)\diff z\right)\\
    &\qquad\quad \times \left(\int\psi(z
      )\bar{\chi}(z'-\alpha_\rmx) e^{-i(z'-\alpha_\rmx/2)\cdot\alpha_\rmp}\diff z' \right)\Bigg]\diff \alpha \\
    &= \int \bar{\phi}(z)\psi(z') e^{i(z-z')\cdot \alpha_\rmp}
      \chi(z-\alpha_\rmx)\bar{\chi}(z'-\alpha_\rmx)\diff z \diff z'\diff \alpha_\rmx\diff \alpha_\rmp \\
    &= (2\pi)^\dm \int \bar{\phi}(z)\psi(z) \left(\int |\chi(z-\alpha_\rmx)|^2 \diff\alpha_\rmx \right)\diff z \\
    &= (2\pi)^\dm \braket{\phi|\psi}
\end{split}
\end{equation}
where to get from the second to the third line we use the Fourier inversion formula.

Then, if $E$ is any trace-class operator, we can write using the singular value
decomposition (which one can obtain from the spectral theorem applied to the 
polar decomposition $E=U\sqrt{E E^\dagger}$)
\begin{equation}
    E = \sum_j \sigma_j \ket{\phi_j}\!\!\bra{\psi_j}
\end{equation}
for some orthonormal bases $\phi_j$, $\psi_j$ on $L^2(\Real^\dm)$.  Here the singular
values $\sigma_j$ are nonnegative and satisfy $\sum_j \sigma_j = \|E\|_1<\infty$.
In this case we can expand the trace using this sum and apply~\eqref{fourier-inv} 
\begin{equation}
\begin{split}    
\trace[E] &= \sum_j \sigma_j \braket{\psi_j|\phi_j} \\
    &=  \sum_j \sigma_j \frac{1}{(2\pi)^{\dm}}\int \braket{\chi_\alpha|\phi_j}\braket{\psi_j|\chi_\alpha}\diff\alpha \\
    &= \frac{1}{(2\pi)^{\dm}} \int \braket{\chi_\alpha|E|\chi_\alpha}\diff\alpha.
\end{split}
\end{equation}
The last line is obtained by swapping the integral and the sum, which can be done
because the sum is absolutely convergent.
\end{proof}

\newtheorem*{lemma:husimi-convolution}{Lemma \ref{lem:husimi-convolution}}
\begin{lemma:husimi-convolution}
    \husimiconvolutionlemma
\end{lemma:husimi-convolution}
\begin{proof}
    We have
    \begin{align}
    \Husimi{\A}(\alpha) 
    &=\braket{\chi_\alpha|\rho|\chi_\alpha} 
    = \trace[\rho(\ket{\chi_\alpha}\!\!\bra{\chi_\alpha})] \\
    &= (2\pi)^{\dm} \int \Wigner{\A}(\beta) \Wigner{D_\alpha\ket{\chi}\!\bra{\chi}D_\alpha^\dagger}(\beta)\diff \beta\\
    &= (2\pi)^{\dm}  \int \Wigner{\A}(\beta) \Wigner{\chi}(\beta-\alpha)\diff \beta
    \end{align}
    where we get the second line from Lemma~\ref{lem:hilbert-schmidt-wigner} and the third line from the fact that the map $\eta\mapsto D_\alpha \eta D_\alpha^\dagger$ on quantum states corresponds to a displacement of the Wigner function by $\alpha$:
    \begin{align}
    \Wigner{D_\alpha \eta D_\alpha^\dagger}(\beta)
    &= \frac{1}{(2\pi)^{2\dm}}\int e^{-i\beta\wedge\xi}\Char{D_\alpha \eta D_\alpha^\dagger}(\xi)\diff \xi\\
    &= \frac{1}{(2\pi)^{2\dm}}\int e^{-i\beta\wedge\xi}\trace[D_\alpha \eta D_{-\alpha} D_\xi]\diff \xi\\
    &= \frac{1}{(2\pi)^{2\dm}}\int e^{-i(\beta-\alpha)\wedge\xi}\trace[\eta D_\xi]\diff \xi\\
    &=\Wigner{\eta}(\beta-\alpha)
    \end{align}
    Furthermore, since $\chi$ is a Schwartz function, so is $\Kernel{\chi}(x,y) = \bar\chi(x)\chi(y)$, and hence by Lemma~\ref{lem:wigner-kernel-schwartz} we have that $\WignerMinus{\chi}(\alpha) = \Wigner{\chi}(-\alpha)$ is a Schwartz function.
\end{proof}

\bibliographystyle{amsxport}
\bibliography{ws}

\end{document}